%% file: main.tex
\documentclass[copyright,creativecommons]{eptcs}
\providecommand{\event}{GandALF 2026} 

\usepackage{iftex}

\ifpdf
  \usepackage{underscore}         
  \usepackage[T1]{fontenc}        
\else
  \usepackage{breakurl}           
\fi

\title{Non-Wellfounded and Cyclic Proofs for LTL:\\ A Syntactic Correspondence with Linear Nested Sequents}
\author{Tim S. Lyon
\institute{TU Dresden \\ Nöthnitzer Straße 46, 01187 Dresden, Germany}
\email{\quad timothy\_stephen.lyon@tu-dresden.de}
\and
Lukas Zenger\thanks{Supported by the Swiss National Science Foundation, Grant P500-2\_235387}
\institute{Peking University\\
5 Yiheyuan Rd., Haidian District,
100871 Beijing,
China}
\email{\quad lukas.zenger@pku.edu.cn}
}
\def\titlerunning{Non-Wellfounded and Cyclic Proofs for LTL}
\def\authorrunning{T.S. Lyon \& L. Zenger}

\input{commands}

\begin{document}
\maketitle

\begin{abstract}
We introduce the formalism of non-wellfounded and cyclic linear nested sequent calculi, developing concrete systems for linear temporal logic ($\ltl$). The paper addresses two central problems, which we call \emph{cycle recognition} and \emph{unraveling}. Cycle recognition concerns identifying cycles in non-wellfounded proofs in order to extract corresponding cyclic proofs, while unraveling studies the converse transformation, from cyclic proofs to non-wellfounded ones. Although these processes are well understood for Gentzen sequents, they have received little attention for more expressive sequent formalisms and become more challenging in the linear nested sequent setting. To address cycle recognition, we show the completeness of non-wellfounded proofs relative to a particular normal form exhibiting a property we call \emph{saturation recurrence}, which enables the systematic extraction of cyclic proofs. To address unraveling, we introduce a specialized procedure that shifts rule applications forward along linear nested sequents, allowing non-wellfounded proofs to be reconstructed from cyclic ones. Overall, our work provides new proof-theoretic techniques for cycle recognition and unraveling in expressive multisequent formalisms.
\end{abstract}

\section{Introduction}

\input{sec-intro}

\section{Logical Preliminaries: Linear Temporal Logic}\label{sec:prelims}

\input{sec-logical-prelims}

\section{Linear Nested Sequents}\label{sec:calculi}

\input{sec-calculus}



\section{Correspondence between Non-Wellfounded and Cyclic Proofs}\label{sec:correspondence}

\input{sec-correspondence}




\section{Concluding Remarks}\label{sec:conclusion}

\input{sec-conclusion}

\bibliographystyle{eptcs}
\bibliography{bibliography}









\end{document}

%% file: commands.tex
\usepackage{amssymb, amsmath, amsthm}
\usepackage{color}
\usepackage{bussproofs}
\usepackage{proof}
\usepackage{mathtools}
\usepackage{relsize}
\usepackage{xspace}
\usepackage{upgreek}
\usepackage{multicol}
\usepackage{stmaryrd}
\usepackage{mathrsfs}
\usepackage{tikz}
\usepackage{graphicx}
\usetikzlibrary{decorations.pathmorphing}

\usepackage{algpseudocode}
\usepackage[ruled,vlined,linesnumbered]{algorithm2e}
\usepackage{mathdots}

\newtheorem{theorem}{Theorem}[section]
\newtheorem{lemma}[theorem]{Lemma}
\newtheorem{definition}[theorem]{Definition}
\newtheorem{example}[theorem]{Example}
\newtheorem{corollary}[theorem]{Corollary}
\newtheorem{remark}[theorem]{Remark}

\newcommand{\case}[1]{%
  \smallskip
  \noindent\textsc{#1}
}

\newcommand{\fig}{Figure}

\newcommand{\iffi}{iff }
\def\phi{\varphi}

\definecolor{tim}{RGB}{0, 0, 250}

\newcommand{\sufo}[1]{\mathrm{sufo}(#1)}

\newcommand{\imp}{\rightarrow}
\newcommand{\lang}{\mathscr{L}}

\newcommand{\atm}{\mathsf{Atm}}

\newcommand{\X}{\mathtt{X}}

\newcommand{\until}{\mathtt{U}}

\newcommand{\ltl}{\mathsf{LTL}}

\newcommand\citet[1]{\citeauthor{#1}~(\citeyear{#1})}

\newcommand{\lnsg}{\mathcal{G}}
\newcommand{\lnsh}{\mathcal{H}}
\newcommand{\lnsk}{\mathcal{K}}
\newcommand{\lnsf}{\mathcal{F}}

\newcommand{\lnsq}{\mathcal{Q}}
\newcommand{\sep}{\sslash}
\newcommand{\id}{\mathsf{id}}
\newcommand{\botl}{\bot\mathsf{L}}

\newcommand{\disr}{\lor \mathsf{R}}
\newcommand{\conr}{\land\mathsf{R}}

\newcommand{\impl}{{\imp}\mathsf{L}}
\newcommand{\impr}{{\imp}\mathsf{R}}
\newcommand{\nextri}{\X \mathsf{R}_{1}}
\newcommand{\nextrii}{\X \mathsf{R}_{2}}
\newcommand{\nextli}{\X \mathsf{L}_{1}}
\newcommand{\nextlii}{\X \mathsf{L}_{2}}

\newcommand{\untilli}{\until \mathsf{L}_{1}}
\newcommand{\untillii}{\until \mathsf{L}_{2}}

\newcommand{\untilri}{\until \mathsf{R}_{1}}
\newcommand{\untilrii}{\until \mathsf{R}_{2}}

\newcommand{\calc}{\mathsf{LNS_{LTL}}}

\newcommand{\calci}{\mathsf{LNS^{\infty}_{LTL}}}
\newcommand{\calccyc}{\mathsf{LNS^{cyc}_{LTL}}}
\newcommand{\sts}{\bar{\sigma}}

\newcommand{\sar}{\vdash}

\newcommand{\prf}{\pi}
\newcommand{\fint}{f}

\newcommand{\tuple}[1]{( #1 )}

\newcommand{\sat}{\vDash}

\newcommand{\der}{\pi} 

\newcommand{\set}[1]{\{#1\}}

\newcommand{\branch}{\mathcal{B}}
\newcommand{\lnspath}{\mathcal{P}}

\newcommand{\comp}[1]{\mathsf{cp}(#1)}

\newcommand{\last}[1]{\mathsf{end}(#1)}

\newcommand{\model}{\bar{\sigma}}

\newcommand{\proves}{\Vdash}

\newcommand{\size}[1]{\left\Vert#1\right\Vert}

\newcommand{\ru}{\mathsf{r}}

\newcommand{\shift}{h} 

\newcommand{\nshole}{\{\cdot\}}

\newcommand{\meq}{=_{M}}

\newcommand{\idxset}{I}

\newcommand{\an}{\mathsf{a}}
\newcommand{\anb}{\mathsf{b}}

\newcommand{\anno}{\mathsf{a}}
\newcommand{\annof}{\mathsf{f}}
\newcommand{\annou}{\mathsf{u}}

%% file: sec-intro.tex

Modal fixed-point and program logics extend modal languages with operators for expressing inductive and coinductive definitions. 
Due to these constructs, such logics often require reasoning principles that go beyond finite, well-founded proofs~\cite{AfsLei17,ForSan13,Min78}. One approach to obtain cut-free and analytic proof systems for such logics is to employ \emph{non-wellfounded proofs}, which take the form of non-wellfounded trees and rely on global soundness conditions to ensure correctness. Such proofs can be viewed as a formalization of reasoning by infinite descent. The relationship between induction, infinite descent, and non-wellfounded proofs has been studied in particular by Brotherston~\cite{Brotherston2006}. He discusses the sense in which induction and infinite descent can be regarded as equivalent proof principles, while emphasizing that the precise correspondence depends on the underlying formalization. When non-wellfounded proofs have the shape of a \emph{regular tree}---containing only finitely many distinct subtrees---they can be folded into finite proofs in which certain leaves are linked back to internal nodes, thereby creating cycles. This gives rise to \emph{cyclic proofs}, which  provide a finite representation of non-wellfounded proofs and are amenable to automated proof-search~\cite{RooZen22} and to computing interpolants~\cite{AfsLei22,Sha14}.

In the setting of \emph{Gentzen sequents} (pairs of (multi)sets of logical formulae), the relationship between non-wellfounded and cyclic proofs is well understood. As alluded to above, if a non-wellfounded proof can be transformed into a regular tree, then it is straightforwardly convertible into a cyclic proof; conversely, one can unravel a cyclic proof into a non-wellfounded proof. By contrast, the correspondence between non-wellfounded and cyclic proofs has received far less attention in the context of \emph{multisequents}.

Multisequents are generalizations of Gentzen sequents, obtained by embedding Gentzen sequents into more complex data structures. Examples include hypersequents (multisets of Gentzen sequents)~\cite{Avr96,Pot83}, linear nested sequents (lines of Gentzen sequents)~\cite{Mas92,Lel15}, and labeled sequents (graphs of Gentzen sequents)~\cite{Sim94,Vig00}. For a survey of multisequent systems and their relationships, see~\cite{LyoEtAl25}. Multisequents were introduced over the past few decades to provide cut-free and analytic sequent-style systems for modal and related logics---particularly for logics where a cut-free Gentzen sequent calculus was not known to exist. 
Multisequent systems often enjoy properties that can make them better suited for proof analysis and applications than traditional Gentzen systems. For instance, many multisequent calculi consist entirely of invertible rules, feature symmetric rules that facilitate cut-elimination, and use structures conducive to counter-model extraction. 

Research on non-wellfounded multisequent calculi remains nascent, with only a handful of works on non-wellfounded hypersequent or labeled sequent systems. A set of non-wellfounded and cyclic hypersequent systems for modal logics with the master modality capturing different frame conditions was provided by Rooduijn~\cite{Roo21,Rooduijn2024}, while Das and Girlando provided a cyclic hypersequent system for transitive closure logic~\cite{DasGir23}. Docherty and Rowe provided a cut-free non-wellfounded labeled system for propositional dynamic logic as well as a cyclic labeled system that includes cut~\cite{DocRow19}.  More recently, Das et al.~\cite{DasGieMar24} presented non-wellfounded labeled calculi for classical and intuitionistic Gödel-Löb logic, with Aguilera and Pacheco subsequently providing a cyclic version of the intuitionistic system~\cite{AguPac25}. Furthermore, Afshari et al.~\cite{AfshariEtAl2023} provided non-wellfounded calculi for intuitionistic temporal logics which employ formula nesting, and their work was improved on by Men\'endez Turata~\cite{MenendezTurata2024}, who obtained a cut-free labeled cyclic proof system. Thus, for multisequents more complex than hypersequents, the works by Aguilera and Pacheco and by Men\'endez Turata serve as the only examples of cut-free cyclic multisequent systems.


Our aim in this paper is to advance this line of research by providing a deeper understanding of cyclic proofs in the context of multisequents. In particular, we focus on non-wellfoundedness and cyclicity within the \emph{linear nested sequent (LNS)} formalism. Linear nested sequents, which consist of ordered lines of Gentzen sequents, generalize hypersequents. This formalism was introduced by Lellmann~\cite{Lel15} and builds on the 2-sequent framework originally developed by Masini~\cite{Mas92,Mas93}. As a case study, we provide cut-free non-wellfounded and cyclic linear nested sequent calculi for linear temporal logic ($\ltl$), which is our first contribution in this paper. We then study and address two central problems, which we call \emph{cycle recognition} and \emph{unraveling}. The cycle recognition problem asks how to identify cycles in non-wellfounded proofs in order to construct a corresponding cyclic proof, whereas the unraveling problem concerns the reverse process: transforming a cyclic proof into a non-wellfounded one.

Detecting cycles in non-wellfounded multisequent proofs is particularly challenging because multisequents employ richer structures than traditional Gentzen sequents. Repetitions do not recur verbatim along infinite branches; instead, structures tend to grow. Consequently, cycles must be identified not as identical repeated sequents, but as recurring growth patterns. We resolve this problem by showing completeness with respect to non‑wellfounded proofs satisfying the \emph{saturation recurrence property}. This guarantees that repetitive growth patterns can be detected along infinite branches, thereby enabling the extraction of cyclic proofs. This constitutes our second contribution.

Unraveling cyclic proofs in this setting is likewise non-trivial since the standard unraveling procedure is inapplicable due to the more complex structures used in proofs. As a third contribution, we show how this problem can be addressed by introducing a specialized unraveling procedure that shifts rule applications forward along linear nested sequents. This makes it possible to reconstruct a non-wellfounded proof in which linear nested sequents exhibit repetitive growth along infinite branches. Since the non-wellfounded system for $\ltl$ is proven sound and complete, the above two transformations yield soundness and completeness of the cyclic system as a corollary.

\smallskip
\noindent
\textbf{LTL and Related Systems.} Linear temporal logic ($\ltl$) is an important program logic introduced by Pnueli~\cite{Pnu77} for reasoning about dynamically evolving systems. 
The logic is equipped with the temporal operators \emph{next} $\X$ and \emph{until} $\until$, which enable the specification of infinite behaviors over execution traces. This gives 
$\ltl$ considerable expressive power while still admitting effective decision procedures. 
Due to these favorable properties, $\ltl$ is widely used in model checking and formal verification, where properties such as safety (e.g., ‘no bad state will be reached’) and liveness (e.g., ‘a good state will eventually be reached’) can be specified and automatically checked.

A variety of sequent calculi have been proposed for $\ltl$. Paech introduced a Gentzen-style system for the logic~\cite{Pae89}; however, this system contains an induction rule that is an instance of cut, and is therefore not cut-free. To address this issue, Brünnler and Lange~\cite{BruLan08} later developed a cut-free cyclic Gentzen calculus that employs annotations---called \emph{histories}---to detect cycles in proofs. A variant of Brünnler and Lange's calculus was discussed by Kokkinis and Studer~\cite{KokStu16} for the unary fragment of $\ltl$, in which weakening is syntactically admissible. More recently, Alonderis et al.~\cite{AloEtAl20} provided a non-annotated variant of the Brünnler-Lange calculus for full $\ltl$. Boretti~\cite{Bor08} introduced an infinitary labeled system for $\ltl$, which relies on an $\omega$-rule with infinitely many premises, rather than on non-wellfounded proofs, to establish validity. On the satisfiability side, Reynolds~\cite{Rey16} provided a one-pass tree-shaped tableau for $\ltl$ satisfiability that employs the $\mathsf{LOOP}$ and $\mathsf{PRUNE}$ rules to detect satisfiable and unsatisfiable cycles in tableau branches, respectively. While the cycle detection underlying Reynolds' tableau is reminiscent of the cycle recognition studied in the present work, the precise relationship between these approaches remains open.

Since our linear nested sequent systems fall within the paradigm of non-wellfounded proof theory, they are most closely related to~\cite{BruLan08,KokStu16,AloEtAl20}. However, they differ in that we employ the richer formalism of linear nested sequents rather than traditional Gentzen sequents, our calculus is \emph{formula-driven} (i.e., it does not contain any structural rules), and all rules are invertible. We remark that our choice of formalism is motivated by two considerations. First, syntactic cut-elimination remains open for all existing non-wellfounded Gentzen calculi for $\ltl$~\cite{BruLan08,KokStu16,AloEtAl20}, and we conjecture that the additional structural expressiveness of linear nested sequents provides a pathway to resolving this problem. Second, our work contributes to the broader goal of developing a general theory of non-wellfounded and cyclic multisequent systems: the cycle recognition and unraveling techniques introduced here serve as a foundation for extending these methods to other expressive formalisms beyond the linear nested setting.

\smallskip
\noindent
\textbf{Outline of Paper.} In Section~\ref{sec:prelims} we introduce the preliminaries for $\ltl$. Section~\ref{sec:calculi} introduces our non-wellfounded and cyclic linear nested sequent systems for $\ltl$ and establishes soundness and completeness of the non-wellfounded calculus. Section~\ref{sec:correspondence} provides our solutions to the cycle recognition and unraveling problems, establishing a syntactic correspondence between non-wellfounded and cyclic linear nested sequent proofs, which entails the soundness and completeness of the cyclic system. Finally, in Section~\ref{sec:conclusion}, we conclude and discuss future work. We note that a full version of this paper, containing all proofs, is available at~\cite{LyoZenFull}.

%% file: sec-logical-prelims.tex
Let $\atm := \{p,q,r,\ldots\}$ be a denumerable set of \emph{(propositional) atoms}. We define the language $\lang$ to be the collection of all formulae $A$ generated via the following grammar in BNF:
$$
A, B ::= p \mid \bot \mid (A \imp B) \mid \X A \mid (A \until B)
$$
with $p \in \atm$. We use $A$, $B$, $C$, $\ldots$ to denote formulae from $\lang$. The formula $\X A$ expresses `$A$ is true in the \emph{next} time-step' while $A \until B$ expresses `$A$ is true \emph{until} $B$ is true (and $B$ eventually becomes true)'. We have opted to use a minimal signature for $\ltl$ to simplify the presentation of our linear nested sequent systems in the subsequent section. We define  $\neg A := A \imp \bot$, $\top := \neg \bot$, $A \lor B := \neg A \imp B$, $A \land B := \neg (A \imp \neg B)$, and $A \leftrightarrow B := (A \imp B) \land (B \imp A)$. Other temporal operators are defined in terms of $\until$ as follows: $\mathtt{F} A := \top \until A$ expresses `$A$ is \emph{eventually} true', $\mathtt{G} A:= \neg \mathtt{F} \neg A$ expresses `$A$ is \emph{henceforth} true', and $A \mathtt{R} B := \neg (\neg A \until \neg B)$ expresses `$A$ \emph{releases} B'. A \emph{complex formula} is a formula of the form $A \imp B$, $\X A$, or $A \until B$. We let $\sufo{A}$ denote the set of all \emph{subformulae} of $A$, defined in the usual way, and call a formula $B$ a \emph{subformula} of $A$ \iffi $B \in \sufo{A}$. For a set of formulae $\Gamma$, we define $\sufo{\Gamma} := \bigcup_{A \in \Gamma} \sufo{A}$.


\begin{definition}\label{def:semantics} A \emph{state sequence} is an infinite sequence of states $\bar\sigma=\tuple{\sigma_0,\sigma_1,\ldots}$, where each state $\sigma_i \subseteq \atm$. The \emph{satisfaction} of a formula $A$ on $\bar\sigma$ at the time point $i \ge0$, written $\bar\sigma, i \sat A$, is defined by:
\begin{itemize}
\item $\bar\sigma, i \sat p$ \iffi  $p \in \sigma_i$;



\item $\bar\sigma, i \not\sat \bot$;

\item $\bar\sigma, i \sat A \imp B$ \iffi $\bar\sigma, i \not\sat A$ or $\bar\sigma, i \sat B$;


\item $\bar\sigma, i \sat \X A$ \iffi $\bar\sigma, i+1 \sat A$;

\item $\bar\sigma, i \sat A \until B$ \iffi there is a $j \geq i$ such that $\bar\sigma, j \sat B$ and for every $i \leq n < j$, $\bar\sigma, n \sat A$;



\item $\bar\sigma \sat A$ \iffi $\bar\sigma, 0 \sat A$.

\end{itemize}
We say that $A$ is \emph{valid}, written $\sat A$, \iffi for each state sequence $\bar\sigma$ we have $\bar\sigma \sat A$; otherwise, a formula $A$ is \emph{invalid}, written $\not\sat A$. We define the logic $\ltl$ to be the set of all valid formulae from $\lang$.
\end{definition}

\begin{example}
  It is straightforward to verify, using the semantic clauses defined above, that for any formulae $A,B \in \lang$, $A \until B \leftrightarrow B \lor (A \land \X (A \until B))$ is valid. The formula $B \lor (A \land \X (A \until B))$ is called the \emph{unfolding} of $A \until B$. 
\end{example}

%% file: sec-calculus.tex
This section introduces linear nested sequents as well as the non-wellfounded and cyclic calculi for $\ltl$. To ensure soundness of the cyclic calculus, we make use of focusing to keep track of good formula traces. 

An \emph{annotated formula} is a tuple $(A, \anno)$ where $A$ is a formula and $\anno \in \{ \annou, \annof \}$ is an \emph{annotation}. The annotation $\annof$ designates that the formula is \emph{in focus} while $\annou$ designates that the formula is \emph{unfocused}. We will write annotated formulae as $A^\anno$ and define $A^{\annou} := A$ for simplicity. That is, we take unfocused formulae to just be regular formulae; this will simplify our presentation. We do not use parentheses when displaying annotated formulae, e.g., $A \until B^\annof$ should be read as $(A \until B)^\annof$.

For a set of annotated formulae $\Gamma$, we define $\Gamma^-:= \{A \mid A^\an \in \Gamma\}$, i.e., $\Gamma^-$ is the set of formulae obtained by removing annotations. A \emph{Gentzen sequent} is an expression of the form $\Gamma \sar \Delta$ such that (1) $\Gamma$ and $\Delta$ are finite sets of annotated formulae, (2) at most one formula in $\Gamma$ is in focus, (3) if a formula in $\Gamma$ is in focus, then it is of the form $A \until B$, and (4) every formula in $\Delta$ is unfocused. We call $\Gamma$ the \emph{antecedent} and $\Delta$ the \emph{consequent} of $\Gamma \sar \Delta$.  If $\Gamma \sar \Delta$ contains a formula in focus, then we call $\Gamma \sar \Delta$ \emph{focused}, and otherwise we call  $\Gamma \sar \Delta$ \emph{unfocused}. 

\begin{figure*}[t]
\noindent\hrule

\begin{center}
\begin{tabular}{c c c}
\AxiomC{$\phantom{\lnsg}$}
\RightLabel{$\id$}
\UnaryInfC{$\lnsg \sep \Gamma, p \sar p, \Delta \sep \lnsh$}
\DisplayProof

&

\AxiomC{$\phantom{\lnsg}$}
\RightLabel{$\botl$}
\UnaryInfC{$\lnsg \sep \Gamma, \bot \sar \Delta \sep \lnsh$}
\DisplayProof
\end{tabular}
\end{center}


\begin{center}
\begin{tabular}{c c}
\AxiomC{$\lnsg \sep \Gamma, A \sar  B, \Delta \sep \lnsh$}
\RightLabel{$\impr$}
\UnaryInfC{$\lnsg \sep \Gamma \sar A \imp B, \Delta \sep \lnsh$}
\DisplayProof

&

\AxiomC{$\lnsg \sep \Gamma, B \sar \Delta \sep \lnsh$}
\AxiomC{$\lnsg \sep \Gamma \sar A, \Delta \sep \lnsh$}
\RightLabel{$\impl$}
\BinaryInfC{$\lnsg \sep \Gamma, A \imp B \sar \Delta \sep \lnsh$}
\DisplayProof
\end{tabular}
\end{center}


\begin{center}
\begin{tabular}{c c c}
\AxiomC{$\lnsg \sep \Gamma \sar \Delta \sep \Sigma, A \sar \Pi \sep \lnsh$}
\RightLabel{$\nextli$}
\UnaryInfC{$\lnsg \sep \Gamma, \X A \sar \Delta \sep \Sigma \sar \Pi \sep \lnsh$}
\DisplayProof

&

\AxiomC{$\lnsg \sep \Gamma \sar \Delta \sep A \sar $}
\RightLabel{$\nextlii$}
\UnaryInfC{$\lnsg \sep \Gamma, \X A \sar \Delta$}
\DisplayProof

&

\AxiomC{$\lnsg \sep \Gamma, A \until B^\annou \sar \Delta \sep \lnsh$}
\RightLabel{$\annou$}
\UnaryInfC{$\lnsg \sep \Gamma, A \until B^\annof \sar \Delta \sep \lnsh$}
\DisplayProof
\end{tabular}
\end{center}


\begin{center}
\begin{tabular}{c c c}

\AxiomC{$\lnsg \sep \Gamma \sar \Delta \sep \Sigma \sar A, \Pi \sep \lnsh$}
\RightLabel{$\nextri$}
\UnaryInfC{$\lnsg \sep \Gamma \sar \X A, \Delta \sep \Sigma \sar \Pi \sep \lnsh$}
\DisplayProof

&

\AxiomC{$\lnsg \sep \Gamma \sar \Delta \sep \sar A$}
\RightLabel{$\nextrii$}
\UnaryInfC{$\lnsg \sep \Gamma \sar \X A, \Delta$}
\DisplayProof

&

\AxiomC{$\lnsg \sep \Gamma, A \until B^\annof \sar \Delta \sep \lnsh$}
\RightLabel{$\annof$}
\UnaryInfC{$\lnsg \sep \Gamma, A \until B^\annou \sar \Delta \sep \lnsh$}
\DisplayProof
\end{tabular}
\end{center}

\begin{center}
\begin{tabular}{c}
\small
\AxiomC{$\lnsg \sep \Gamma, B \sar \Delta \sep \Sigma \sar \Pi \sep \lnsh$}
\AxiomC{$\lnsg \sep \Gamma, A \sar \Delta \sep \Sigma, A \until B^\an  \sar \Pi \sep \lnsh$}
\RightLabel{$\untilli$}
\BinaryInfC{$\lnsg \sep \Gamma, A \until B^\an \sar \Delta \sep \Sigma \sar \Pi \sep \lnsh$}
\DisplayProof
\end{tabular}
\end{center}

\begin{center}
\begin{tabular}{c}
\AxiomC{$\lnsg \sep \Gamma \sar A, B, \Delta \sep \Sigma \sar \Pi \sep \lnsh$}
\AxiomC{$\lnsg \sep \Gamma \sar B, \Delta \sep \Sigma \sar A \until B, \Pi \sep \lnsh$}
\RightLabel{$\untilri$}
\BinaryInfC{$\lnsg \sep \Gamma \sar A \until B, \Delta \sep \Sigma \sar \Pi \sep \lnsh$}
\DisplayProof
\end{tabular}
\end{center}


\begin{center}
\begin{tabular}{c c}
\AxiomC{$\lnsg \sep \Gamma, B \sar \Delta$}
\AxiomC{$\lnsg \sep \Gamma, A \sar \Delta \sep A \until B^\an \sar $}
\RightLabel{$\untillii$}
\BinaryInfC{$\lnsg \sep \Gamma, A \until B^\an \sar \Delta$}
\DisplayProof

&

\AxiomC{$\lnsg \sep \Gamma \sar A, B, \Delta$}
\AxiomC{$\lnsg \sep \Gamma \sar B, \Delta \sep \sar A \until B$}
\RightLabel{$\untilrii$}
\BinaryInfC{$\lnsg \sep \Gamma \sar A \until B, \Delta$}
\DisplayProof
\end{tabular}
\end{center}

\noindent\hrule
\caption{Inference Rules $\calc$.\label{fig:calc}}
\end{figure*}

A \emph{linear nested sequent (LNS)} is an expression of the form $\lnsg := \Gamma_{0} \sar \Delta_{0} \sep \cdots \sep \Gamma_{n} \sar \Delta_{n}$ such that (1) $\Gamma_{i} \sar \Delta_{i}$ is a Gentzen sequent for $0 \leq i \leq n$ called the \emph{i-component} of $\lnsg$ (or simply \emph{component} of $\lnsg$ if the index $i$ is not of importance), and (2) 
at most one component of $\lnsg$ is focused. LNSs are particularly natural for LTL, where every state in a state sequence has a unique successor. Intuitively, an LNS may thus be regarded as a finite `snapshot' of a state sequence, with each Gentzen sequent corresponding to one state and its right neighbor sequent corresponding to its successor state. We use $\lnsg$, $\lnsh$, $\lnsk$ and so on to denote LNSs. For an LNS $\lnsg = \Gamma_{0} \sar \Delta_{0} \sep \cdots \sep \Gamma_{n} \sar \Delta_{n}$ we refer to the $(n{-}1)$-component as the \emph{penultimate component}, to the $n$-component as the \emph{end component}, and to any $k$-component such that $k < n$ as an \emph{interior component}, that is, an interior component is any component that is \emph{not} the end component. 


We define the multiset of components of a linear nested sequent $\lnsg = \Gamma_{0} \sar \Delta_{0} \sep \cdots \sep \Gamma_{n} \sar \Delta_{n}$ to be $\comp{\lnsg} := \set{\Gamma_{i} \sar \Delta_{i} \mid 0 \leq i \leq n}$. The \emph{length} of an LNS $\lnsg$ is the number of components it contains, i.e., $\size{\lnsg} := |\comp{\lnsg}|$. For an LNS $\lnsg$, we let $\lnsg(i)$ be the $i$-component of $\lnsg$, if $0 \leq i < \size{\lnsg}$, and we let $\last{\lnsg} := \lnsg(\size{\lnsg}-1)$ denote the end component of $\lnsg$. For a Gentzen sequent $\Gamma \sar \Delta$, let $\sufo{\Gamma \sar \Delta} := \sufo{\Gamma^{-} \cup \Delta^{-}}$. For an LNS $\lnsg = \Gamma_0 \sar \Delta_0 \sep \cdots \sep \Gamma_n \sar \Delta_n$, we define $\sufo{\lnsg}$ as follows: $\sufo{\lnsg} := \bigcup_{0 \leq i \leq n} \sufo{\Gamma_i^{-} \cup \Delta_i^{-}}.$ We can interpret LNSs by means of their \emph{formula interpretation}, defined below:
$$
\fint(\Gamma \sar \Delta) := \bigwedge \Gamma^- \imp \bigvee \Delta^-
\quad
\fint(\Gamma \sar \Delta \sep \lnsg) := \bigwedge \Gamma^- \imp (\bigvee \Delta^- \lor \X \fint(\lnsg))
$$

\begin{definition}\label{def: satisfaction LNS} Let $\lnsg$ be an LNS, $0 \leq i < \size{\lnsg}$, and $\Bar{\sigma}$ be a state sequence. We define $\model, i \sat \lnsg(i)$ \iffi $\model, i \sat \bigwedge \Gamma^- \to \bigvee \Delta^-$ with $\lnsg(i) = \Gamma \sar \Delta$. We define $\model \sat \lnsg$ \iffi $\model \sat \fint(\lnsg)$, i.e., there exists $0 \leq i < \size{\lnsg}$ such that $\model, i \sat \lnsg(i)$.\footnote{Note that the annotations do not convey any semantic meaning.} An LNS $\lnsg$ is \emph{valid} \iffi for each state sequence $\Bar{\sigma}$, $\model \sat \lnsg$; otherwise, $\lnsg$ is \emph{invalid}.
\end{definition}

We define $\calc$ to be the set of all rules depicted in \fig~\ref{fig:calc}. The $\id$ and $\botl$ rules are \emph{initial rules} and we call the conclusion of an initial rule an \emph{initial sequent}. The \emph{focusing rules} $\annou$ and $\annof$, read bottom-up, unfocus and focus an until-formula, respectively.  The remaining rules are \emph{logical rules} that introduce complex formulae into either the antecedent or consequent of a component of an LNS. As discussed in more detail below, the $\untilli$ and $\untillii$ rules utilize our focusing mechanism to ensure the soundness of non-wellfounded and cyclic proofs. Moreover, we remark that the $\untilli$, $\untillii$, $\untilri$, and $\untilrii$ rules rewrite an until-formula $A \until B$ into its equivalent unfolding $B \vee (A \wedge \X(A \until B))$. 
The rules $\untilli$ and $\untilri$ operate on until-formulae in an interior component: in the right premise, $A\until B$ is pushed into the next component. The rules $\untillii$ and $\untilrii$ work similarly but operate on until-formulae in the end component, thereby introducing a new component in the right premise.

\begin{remark}\label{remark: annotations}
    The presence of annotations limits the possible instances of rules, due to the restriction that components of an LNS contain at most one focused formula and LNSs contain at most one focused component. Therefore, every component of the conclusion of an instance of $\annof$ must be unfocused, as otherwise the premise would either contain two focused components or a component with two formulae in focus. Similarly, one can verify that every component in the left premise of an until-rule is unfocused. 
\end{remark}

The \emph{principal formula} of a logical rule in $\calc$ is the formula explicitly displayed in the conclusion, and we define the \emph{auxiliary formulae} to be those explicitly displayed in the premises. For example, $A \until B$ is principal in $\untilri$ and $A$, $B$, and $A \until B$ are auxiliary. All other formulae are called \emph{side formulae}. The \emph{principal component} (\emph{auxiliary components}) is (are) the component(s) where the principal (auxiliary, resp.) formula(e) occur. Since we are working with sequents based on sets of formulae (where contraction is implicit), rule instances are either preserving or succinct. A rule instance is \emph{preserving} if the principal formula also occurs as a side formula in the premise, and \emph{succinct} otherwise. For example, the following shows a preserving instance of $\impr$ on the left and a succinct instance on the right:
\begin{center}
\begin{tabular}{c c}
\AxiomC{$\lnsg \sep  A \sar  B, A \imp B \sep \lnsh$}
\RightLabel{$\impr$}
\UnaryInfC{$ \lnsg \sep \sar A \imp B \sep \lnsh$}
\DisplayProof

&

\AxiomC{$\lnsg \sep  A \sar  B \sep \lnsh$}
\RightLabel{$\impr$}
\UnaryInfC{$\lnsg \sep  \sar A \imp B \sep \lnsh$}
\DisplayProof
\end{tabular}
\end{center}


The following lemma, which states that all rules of $\calc$ are locally sound, is readily verified by means of the formula interpretation.

\begin{lemma}\label{lem: local soundness}
    Let $\model$ be a state sequence and consider an instance of a rule  $\mathsf{r}$ of $\calc$ with conclusion $\lnsg$ and premises $\lnsh_1, \ldots, \lnsh_n$. If $\model \not \sat \lnsg$, then there exists $1 \leq i \leq n$ with $\model \not \sat \lnsh_i$.
\end{lemma}

\smallskip
\noindent
\textbf{Non-Wellfounded Proofs.} We define a \emph{derivation} $\der$ of an LNS $\lnsg$ in $\calc$ to be a (potentially infinite) tree of LNSs such that (1) $\lnsg$ is the root and (2) every parent node is the conclusion of a rule of $\calc$ with its children the corresponding premises. An \emph{index set} is a set $\idxset \subseteq \mathbb{N}$ such that (1) $0 \in \idxset$ and (2) if $n+1 \in \idxset$, then $n \in \idxset$, i.e., an index set is an initial segment of the natural numbers, or the whole set of naturals. We use index sets to define certain kinds of (potentially infinite) sequences below. 

A \emph{path} in a derivation $\der$ is a (potentially infinite) sequence $\lnspath = \set{\lnsg_{i}}_{i \in \idxset}$ of LNSs (where $\idxset$ is an index set) such that for each $i \in \idxset$ if $i+1 \in \idxset$, then the LNS $\lnsg_{i}$ is the parent of $\lnsg_{i+1}$. A path $\lnspath = \set{\lnsg_{i}}_{i \in \idxset}$ is \emph{finite} if $\idxset \not = \mathbb{N}$ and \emph{infinite} otherwise. Moreover, $\lnspath = \set{\lnsg_{i}}_{i \in \idxset}$ is \emph{maximal} if $\lnspath$ is finite and $\lnsg_{\lvert \idxset \rvert -1}$ is a leaf or if $\lnspath$ is infinite. A \emph{branch} in a derivation $\der$ is a maximal path $\branch = \set{\lnsg_{i}}_{i \in \idxset}$ of LNSs such that $\lnsg_{0}$ is the root. The \emph{height} of a derivation is defined in the usual way as the maximal length of a branch in the derivation, which may be infinite. A path $\lnspath := \set{\lnsg_{i}}_{i \in \idxset}$ has a \emph{continuous focus trace} \iffi for every $i \in \idxset$, $\lnsg_{i}$ contains a formula in focus. The following lemma follows directly from this definition and Remark~\ref{remark: annotations}:

\begin{lemma}
If a path $\lnspath$ has a continuous focus trace, then neither $\annou$ nor $\annof$ are applied along the path.
\end{lemma}

An infinite path $\lnspath$ in a derivation $\der$ is called \emph{progressing} if (1) $\lnspath$ has a continuous focus trace and (2) $\lnspath$ passes infinitely often through the right premise of an instance of a rule $\mathsf{r} \in \{\untilli, \untillii\}$ in which the formula in focus is principal. An infinite branch $\branch$ in $\der$ is called \emph{progressing} if it has a suffix which is a progressing path.

\begin{definition}
A \emph{non-wellfounded proof} (or, \emph{$\calci$-proof}) of an LNS $\lnsg$ is a derivation $\der$ of $\lnsg$ such that (1) every leaf of $\der$ is an initial sequent and (2) every infinite branch is progressing.
\end{definition}
We let $\calci$ be the non-wellfounded LNS calculus obtained by letting the set of provable sequents be determined by non-wellfounded proofs and write $\calci \proves \lnsg$ if $\lnsg$ has a non-wellfounded proof. \smallskip

\noindent \textbf{Cyclic Proofs.} Let $\lnsg$ be an LNS with $\size{\lnsg} = n+1$ and $\lnsg(i) = \Gamma_{i} \sar \Delta_{i}$. For $0 \leq i \leq n$, we define $\lnsg$ to be \emph{$i$-saturated} \iffi it is not an initial sequent and $\set{A \mid A \in (\Gamma_{i} \cup \Delta_{i})^-} \cap \set{\X A, A \imp B, A \until B \mid A, B \in \lang} = \emptyset$. We call $\lnsg$ \emph{saturated} \iffi for each $0 \leq i < n$, $\lnsg$  is $i$-saturated. Intuitively, saturated sequents serve as checkpoints: since their interior components are free of complex formulae (meaning, no rules are bottom-up applicable to these components), two saturated sequents with matching end components represent identical `states' in a proof, making them natural candidates for the endpoints of a cycle. This serves as the motivation for our notion of cyclic proof given below.

\begin{definition}\label{dfn:cyclic-proof}
A tuple $\tuple{\der,L,c}$ is a \emph{cyclic proof} (or, \emph{$\calccyc$-proof}) of an LNS $\lnsg$ \iffi it satisfies: 
\begin{itemize}

\item[$(1)$] $\der$ is a finite derivation of $\lnsg$;

\item[$(2)$] $L$ is a set containing only leaves of $\der$, which we call \emph{buds};

\item[$(3)$] $c$ is a function mapping each $\lnsg \in L$ to an LNS $c(\lnsg)$ strictly below $\lnsg$, called the \emph{companion} of $\lnsg$, such that $\size{c(\lnsg)} < \size{\lnsg}$, the path from $c(\lnsg)$ to $\lnsg$ has a continuous focus trace and goes through the right premise of a rule $\ru \in \set{\untilli,\untillii}$ with the principal formula in focus, both $c(\lnsg)$ and $\lnsg$ are saturated, and $\last{c(\lnsg)} = \last{\lnsg}$;\footnote{By ``\emph{strictly below $\lnsg$},'' we mean that $c(\lnsg)$ occurs on the path from the root of $\der$ to the leaf $\lnsg$ and is distinct from $\lnsg$.}


\item[$(4)$] every other leaf in $\der$ that is not in $L$ is an initial sequent.

\end{itemize}
\end{definition}
We let $\calccyc$ be the cyclic LNS calculus obtained by letting the set of provable sequents be determined by cyclic proofs and write $\calccyc \proves \lnsg$ if $\lnsg$ has a cyclic proof.

\begin{example} An example of a cyclic proof $(\der, L, c)$ is displayed in \fig~\ref{fig:example-cyc-proof} where $L = \set{\lnsg}$ and $c(\lnsg) = \lnsh$ such that $\lnsg := p \sar q \sep p \until q^\annof \sar p \until q$ and $\lnsh := p \until q^\annof \sar p \until q$. It is straightforward to verify that $c$ defines a correct cycle since $\size{\lnsh} < \size{\lnsg}$, the branch from $\lnsh$ to $\lnsg$ has a continuous focus trace and goes through the right premise of $\untillii$ with the principal formula in focus, both $\lnsg$ and $\lnsh$ are saturated, and $\last{\lnsh} = \last{\lnsg}$.
\end{example}

\begin{figure}
\begin{center}
\begin{tabular}{c c}
$\prf_{1} \ :=$
&
\AxiomC{}
\RightLabel{$\id$}
\UnaryInfC{$q \sar p, q$}
\AxiomC{}
\RightLabel{$\id$}
\UnaryInfC{$q \sar q \sep \sar p \until q$}
\RightLabel{$\untilrii$}
\BinaryInfC{$q \sar p \until q$}
\DisplayProof
\end{tabular}
\end{center}

\begin{center}
\begin{tikzpicture}
\node[] (a) [] {\AxiomC{$\prf_{1}$}

\AxiomC{}
\RightLabel{$\id$}
\UnaryInfC{$p \sar p, q \sep p \until q^\annof \sar$}

\AxiomC{$p \sar q \sep p \until q^\annof \sar p \until q$}
\RightLabel{$\untilri$}
\BinaryInfC{$p \sar p \until q \sep p \until q^\annof \sar$}
\RightLabel{$\untillii$}
\BinaryInfC{$p \until q^\annof \sar p \until q$}
\DisplayProof};


\draw[->,rounded corners=.25cm,dashed]
  (2.8,.75) -- (2.8,1) -- (6.2,1) -- node[right,midway]{$c$} 
  (6.2,-.65) -- (-.2,-.65);

\end{tikzpicture}
\end{center}

\caption{Example of a cyclic proof.\label{fig:example-cyc-proof}}
\end{figure}

\subsection{Soundness and Completeness of $\calci$}


We finish this section by arguing that $\calci$ is sound and complete. The detailed proofs can be found in the full version of this paper~\cite{LyoZenFull}. In the subsequent section, we will define bi-directional proof transformations between non-wellfounded and cyclic proofs, yielding soundness and completeness of $\calccyc$ as a corollary. 

To establish soundness of $\calci$ we proceed via a proof by infinite descent, as is often done for non-wellfounded proof systems, see e.g.~\cite{AfshariEtAl2023}. We assume toward a contradiction that some invalid LNS $\lnsg$ has a non-wellfounded proof $\pi$. Since $\lnsg$ is invalid, there exists a state sequence $\model$ such that $\model \not \sat \lnsg$. We then show how to obtain a branch $\branch =\{\lnsg_i\}_{i\in \mathbb{N}}$ of $\pi$ such that $\model \not\sat \lnsg_i$ for each $i \in \mathbb{N}$, by using the fact that all rules of $\calci$ are locally sound (see Lemma~\ref{lem: local soundness}). Note that $\branch$ must be an infinite branch of $\pi$, since initial sequents are always valid. Simultaneously, we assign a well-founded measure in the form of natural numbers to formulae in focus occurring in LNSs along $\branch$ and show that whenever a formula $A \until B^\annof$ is principal in $\untilli$ or $\untillii$ and $\branch$ passes through the right premise, the measure assigned to $A \until B^\annof$ strictly decreases, and weakly decreases otherwise. Finally, since $\prf$ is a proof, $\branch$ must contain a progressing suffix, meaning that there is a continuous focus trace and the suffix passes infinitely often through the right premise of instances of $\untilli$ and $\untillii$ with the formula in focus principal. Thus, we obtain an infinite decreasing sequence of natural numbers which strictly decreases infinitely often, contradicting the well-foundedness of 
$\mathbb{N}$. This gives the desired contradiction, implying the soundness of $\calci$. 

\begin{theorem}[Soundness of $\calci$]\label{thm:soundness}
    If $\calci \proves \lnsg$, then $\lnsg$ is valid.
\end{theorem}


Completeness of $\calci$ is established by a proof-search argument. Given an LNS $\lnsg$, we build a structure $\prf$, called a \emph{proof-search tree}, by exhaustively applying rules bottom-up to formulae in the interior components of LNSs until every open leaf is either initial or saturated. Initial leaves are closed, as are saturated leaves that are not the conclusion of any rule instance. For any other saturated leaf, we distinguish the following cases. (1) If the end component does not contain a formula of the form $\X A$ or $A \until B^\anno$, then we exhaustively apply implication rules to decompose the formulae in the end component until we obtain an initial sequent or the end component is saturated; in either case the leaf is closed. (2) Otherwise, we apply a rule that adds a new component. In case (2) we also need to apply the focusing rules if there are formulae of the form $A \until B$ present. Since there are in general multiple formulae present that may be focused or unfocused, we introduce so-called \emph{branching focus rules} which have one premise for each possible application of $\annou$ or $\annof$. If we apply a rule that adds a new component in case (2), then the premises are in general not saturated, and so the process can be repeated indefinitely. This results in an infinite tree of LNSs called a proof-search tree for $\lnsg$. Since we employ branching focus rules, this proof-search tree is not a derivation. However, we can obtain derivations from proof-search trees by choosing a single premise for each branching focus rule instance and deleting the other branches. Derivations obtained in that way satisfy the \emph{saturation recurrence property (SRP)}:

\begin{definition} A derivation $\pi$ satisfies the \emph{saturation recurrence property (SRP)} if and only if, for any infinite branch $\branch := \{\lnsg_i\}_{i \in \mathbb{N}}$ in $\pi$, infinitely many $\lnsg_i$ are saturated.
\end{definition}
 
We then show that given a proof-search tree for an LNS $\lnsg$, we can either extract a derivation from it which is a non-wellfounded proof of $\lnsg$, or the proof-search tree contains a non-progressing branch which can be used to construct a counter-model for $\lnsg$. Since valid LNSs do not have counter-models, we obtain that every valid LNS has a non-wellfounded proof.

\iffalse
We call $\nextlii$, $\nextrii$, $\untillii$ and $\untilrii$ \emph{expansion} rules and all other logical rules \emph{length-preserving}. An LNS $\lnsg$ is called \emph{fully saturated} \iffi for $0 \leq i < \size{\lnsg}$, $\lnsg$ is $i$-saturated, i.e., $\lnsg$ is not the conclusion of any instance of a rule in $\calci$. We formalize proof-search by means of \emph{proof-search trees} which represent a systematic search for a proof of a given LNS. 

\begin{definition} A \emph{proof-search tree} for an LNS $\lnsg$ is a derivation $\pi$ of $\lnsg$ such that the following hold:
\begin{enumerate}

\item[$(1)$] Every rule instance in $\pi$ is succinct;

\item[$(2)$] The conclusion of every instance of an expansion rule is saturated;

\item[$(3)$] Every leaf of $\pi$ is an initial sequent or a fully saturated sequent.

\end{enumerate}
\end{definition}

Intuitively, given an LNS $\lnsg$, we obtain a proof-search tree by applying length-preserving rules bottom-up until every leaf is either an initial sequent or saturated. Branches ending in initial sequents or fully saturated leaves are closed, while branches ending in saturated leaves that are \emph{not fully saturated} are extended by bottom-up applying an expansion rule. Since the premise of an expansion rule may no longer be saturated, this process can be repeated. We thus obtain the following lemma, whose proof is standard and omitted.

\begin{lemma}
    Every LNS $\lnsg$ has a proof-search tree.
\end{lemma}

Since in a proof-search tree every rule is applied succinctly and expansion rules are only applied to saturated LNSs, a simple inductive argument shows that every infinite branch tree contains infinitely many saturated sequents and so that every proof-search tree satisfies the SRP. The proof of the following lemma is deferred to Appendix~\ref{app:completeness}.
 
\begin{lemma}\label{lem: SRP}
    Let $\lnsg$ be an LNS and $\pi$ be a proof-search tree for $\lnsg$. Then, $\pi$ satisfies the SRP.
\end{lemma}

Observe that if $\branch=\set{\lnsg_i}_{i \in \mathbb{N}}$ is an infinite branch of a proof-search tree $\pi$, then there are infinitely many $i \in \mathbb{N}$ such that $\size{\lnsg_i} < \size{\lnsg_{i+1}}$, since $\pi$ satisfies the SRP and rules are applied succinctly. In particular if $\lnsg_i$ is saturated and the conclusion of an instance of $\untilrii$ with $\lnsg_{i+1}$ the left premise, then $\lnsg_{i+1}$ is saturated and its size does not grow. However, since the instance is succinct and sequents are ordered pairs of \emph{finite} sets of annotated formulae, this situation can only occur finitely many times before $\branch$ passes through an expansion rule that does increase the size of the premise. We now show that every LNS has an $\calci$-proof which satisfies the SRP or is falsifiable, from which completeness readily follows.

\begin{lemma}\label{lem: proof or falsifiable}
    Let $\lnsg$ be an LNS. Then, $\lnsg$ has an $\calci$-proof which satisfies the SRP or there exists a state sequence $\Bar{\sigma}$ such that $\Bar{\sigma} \not \sat \lnsg$.
\end{lemma}
\begin{proof}
    Let $\lnsg$ be an LNS and let $\pi$ be a proof-search tree for $\lnsg$. Recall that $\pi$ is a derivation. If $\pi$ is a proof, then $\lnsg$ has an $\calci$-proof which satisfies the SRP by Lemma~\ref{lem: SRP}. Otherwise $\pi$ contains a `bad' branch $\branch$: either $\branch$ is finite and ends in a fully saturated sequent or $\branch$ is infinite and not progressing. We only consider the latter case, as the former is argued similarly. 

    So suppose $\branch := \set{\lnsg_{i}}_{i \in \mathbb{N}}$ is an infinite branch of $\der$ which is \emph{not} progressing. We use the branch $\branch$ to construct a state sequence $\sts$ such that $\sts \not\sat \lnsg$. Let us define:
$$
\Gamma_{i} := \!\!\!\!\!\!\!\!\!\!\!\!\!\!\!\! \bigcup_{j \in \mathbb{N}, \, \lnsg_{j}(i) \, = \, (\Gamma \, \sar \, \Delta)} \!\!\!\!\!\!\!\!\!\!\!\!\!\!\!\! \Gamma
\qquad
\Delta_{i} := \!\!\!\!\!\!\!\!\!\!\!\!\!\!\!\! \bigcup_{j \in \mathbb{N}, \, \lnsg_{j}(i) \, = \, (\Gamma \, \sar \, \Delta)} \!\!\!\!\!\!\!\!\!\!\!\!\!\!\!\! \Delta
$$
We define $\sigma_{i} = \set{p \mid p \in \Gamma_{i}}$ and $\bar\sigma=\tuple{\sigma_0,\sigma_1,\ldots}$. We now prove the following by a mutual induction on the length of $A$: (1) if $A^\an \in \Gamma_{i}$ for an annotation $\an$, then $\sts, i \sat A$ and (2) if $A^\an \in \Delta_{i}$ for an annotation $\an$, then $\sts, i \not\sat A$. We consider the case for $A = B \until C$. The full proof is shown in Appendix~\ref{app:completeness}.

\case{Case for $A = B \until C^\an \in \Gamma_{i}$.}By definition, there exists a $j$ such that $B \until C^\an \in \Sigma$ with $\lnsg_{j}(i) = \Sigma \sar \Pi$.  Since $\pi$ satisfies the SRP, there exists an instance of $\untilli$ or $\untillii$ with $B \until C^\an$ principal and $B,C$ and $B \until C^\anb$ residual, where $\anb = \an$ in case $\an \not = \epsilon$. Moreover, since $\branch$ is not progressing, there exists a natural number $k \geq 1$ such that $\untilli$ and $\untillii$ are applied $k$ times above $\lnsg_j$ in the branch $\branch$ such that (1) $B \until C^\anb$ is principal in each instance, (2) for $k-1$ applications, the branch $\branch$ goes through the right premise of $\untilli$ or $\untillii$, and (3) for the $k^\text{th}$ application, the branch $\branch$ goes through the left premise of $\untilli$ or $\untillii$. Otherwise, by the SRP, the path starting at $\lnsg_j$ would have a stable annotation, namely $\anb$, and the path would pass infinitely often through the right premise of $\untilli$ or $\untillii$ with $B \until C^\anb$ principal, which would imply that the path and thus the branch is progressing. Therefore, for $0 \leq n < k-1$, we have $B \in \Gamma_{i+n}$ and $C \in \Gamma_{i+k-1}$. By IH, $\sts, i+n \sat B$ and $\sts, i+k-1 \sat C$, which implies that $\sts, i \sat B \until C$.

\case{Case for $A = B \until C \in \Delta_{i}$.}By definition, there exists a $j$ such that $B \until C \in \Pi$ with $\lnsg_{j}(i) = \Sigma \sar \Pi$. There are two cases to consider. First, it could be the case that $\untilri$ and $\untilrii$ are applied infinitely often above $\lnsg_{j}$ along $\branch$ with $B \until C$ principal and with $\branch$ always going through the right premise of $\untilri$ or $\untilrii$. In this case, we have that $C \in \Delta_{k}$ for all $k \geq i$. By IH, we have that $\sts, k \not\sat C$ for all $k \geq i$, meaning, $\sts, i \not\sat B \until C$. Second, it could be the case that $\untilri$ and $\untilrii$ are applied $k \geq 1$ times\footnote{We note that $k > 0$ since $\pi$ satisfies the SRP by the same argument as before.} above $\lnsg_{j}$ in the branch $\branch$ such that $B \until C$ is principal, the branch $\branch$ goes through the right premise of $\untilri$ or $\untilrii$ for $k-1$ applications, and the branch $\branch$ goes through the left premise of $\untilri$ or $\untilrii$ for the $k^\text{th}$ application. Therefore, for $0 \leq n < k-1$, we have $C \in \Delta_{i+n}$ and $B,C \in \Delta_{i+k-1}$. By IH, $\sts, i+n \not\sat C$, $\sts, i+k-1 \not\sat B$, and $\sts, i+k-1 \not\sat C$, which implies that $\sts, i \not\sat B \until C$.

Finally, let $\lnsg(i) = \Sigma_{i} \sar \Pi_{i}$ for $0 \leq i < \size{\lnsg}$. As $\Sigma_{i} \subseteq \Gamma_{i}$ and $\Pi_{i} \subseteq \Delta_{i}$ for $0 \leq i < \size{\lnsg}$, it follows that $\sts, i \not\sat \lnsg(i)$ for all $0 \leq i < \size{\lnsg}$. Hence, by Definition~\ref{def: satisfaction LNS}, it follows that $\Bar{\sigma} \not \sat \lnsg$ as claimed.
\end{proof}
\fi

\begin{theorem}[Completeness of $\calci$]\label{thm:completeness-non-wf}
If $\lnsg$ is valid, then $\lnsg$ has an $\calci$-proof satisfying the SRP.
\end{theorem}

%% file: sec-correspondence.tex
We now establish bi-directional proof transformations between non-wellfounded and cyclic proofs. We first address cycle recognition in non-wellfounded proofs. In order to turn non-wellfounded proofs into cyclic proofs, we will make use of the SRP, which guarantees that infinite branches contain suitable `repetitions' where they can be pruned and replaced by cycles.

Transforming cyclic proofs into non-wellfounded ones is non-trivial in the LNS setting. The key idea is captured by the Shifting Lemma (Lemma~\ref{lem:shifting}). Given a cyclic proof $(\der,L,c)$ whose conclusion $c(\lnsg)$ is the companion of a bud $\lnsg$, we construct a new cyclic proof of $\lnsg$ by shifting inferences forward in $(\der,L,c)$, yielding a cyclic proof $(\der',L',c')$. Iterating this construction allows one to unravel cycles ad infinitum, thereby producing a non-wellfounded derivation from a cyclic one. By carefully handling annotations in the construction, we can then show that the derivation is a non-wellfounded proof which satisfies the SRP. Together, these bi-directional transformations establish the soundness and completeness of $\calccyc$.

\subsection{From Non-Wellfounded Proofs to Cyclic Proofs}

This section shows how non-wellfounded proofs which satisfy the SRP can be transformed into cyclic proofs. We call a rule $\mathsf{r} \in \set{\nextlii, \nextrii, \untillii, \untilrii}$ an \emph{expansion rule} (cf.~\cite{LyoOst26}). We first show that every infinite branch in a non-wellfounded proof passes infinitely often through an expansion rule. We then show that every infinite branch in a non-wellfounded proof which satisfies the SRP contains a suitable `repetition' and finally use these two results to establish the translation.

\begin{lemma}\label{lem: nwf branches grow}
If $\pi$ is an $\calci$-proof and $\branch$ is an infinite branch of $\pi$, then $\branch$ passes infinitely often through $\nextlii$, $\nextrii$, or through the right premise of $\untillii$ or $\untilrii$.
\end{lemma}
\begin{proof}
 Suppose otherwise and let $\pi$ be an $\calci$-proof and $\branch = \set{\lnsg_i}_{i \in \mathbb{N}}$ be an infinite branch of $\pi$. Then, there exists a $j \in \mathbb{N}$ such that the suffix $\lnspath=\set{\lnsg_i}_{i \geq j}$ of $\branch$ starting at $\lnsg_j$ never passes through instances of $\nextlii$, $\nextrii$ or the right premises of $\untillii$ or $\untilrii$. Therefore, $\size{\lnsg_i} = \size{\lnsg_k}$ for all $i,k \geq j$, which immediately implies that $\branch$ does not have a progressing suffix and is thus not progressing. This contradicts our assumption that $\pi$ is a proof.
\end{proof}

\begin{lemma}\label{lem: suitable repetitions exist} Let $\prf$ be a non-wellfounded proof which satisfies the SRP and let $\branch$ be an infinite branch. Then, $\branch$ contains two LNSs $\lnsg$ and $\lnsh$ such that $\lnsg$ occurs strictly below $\lnsh$ in $\branch$, $\size{\lnsg} < \size{\lnsh}$, the path from $\lnsg$ to $\lnsh$ contains a continuous focus trace and goes through the right premise of $\mathsf{r} \in \set{\untilli, \untillii}$ with the formula in focus principal, and $\last{\lnsg} = \last{\lnsh}$.
\end{lemma}

\begin{proof} Let $\pi$ be an $\calci$-proof of $\lnsf$ satisfying the SRP. All rules of $\calci$ are \emph{analytic}, i.e., if $\lnsk$ is the conclusion and $\lnsk'$ a premise of a rule, then $\sufo{\lnsk'} \subseteq \sufo{\lnsk}$.  Thus, there are at most $2^{3 \times |\sufo{\lnsf}|}$ many Gentzen sequents that can occur as components of an LNS in $\prf$. Since $\pi$ satisfies the SRP, every infinite branch encounters infinitely many saturated sequents. Since there are at most $2^{3 \times |\sufo{\lnsf}|}$ many Gentzen sequents that can serve as an end component, every infinite branch must contain infinitely many saturated sequents with identical end components. Moreover, since $\pi$ is a proof, every infinite branch $\branch$ is progressing and thus contains a suffix with a continuous focus trace that passes infinitely often through the right premise of $\untilli$ or $\untillii$ with the formula in focus principal. Finally, by Lemma~\ref{lem: nwf branches grow}, for every infinite branch $\branch = \{\lnsg_i\}_{i \in \mathbb{N}}$ there are infinitely many $i \in \mathbb{N}$ such that $\size{\lnsg_i} < \size{\lnsg_{i+1}}$. Thus, every infinite branch must contain a pair of saturated sequents $\lnsg$ and $\lnsh$ such that $\size{\lnsg} < \size{\lnsh}$, the path from $\lnsg$ to $\lnsh$ contains a continuous focus trace and goes through the right premise of $\mathsf{r} \in \set{\untilli, \untillii}$ with the formula in focus principal, and $\last{\lnsg} = \last{\lnsh}$. 
\end{proof}

Lemma~\ref{lem: suitable repetitions exist} implies that every infinite branch of a proof $\pi$ which satisfies the SRP contains a suitable candidate for a cycle, which allows us to extract a cyclic proof. 

\begin{theorem}\label{thm:non-wf-to-cyclic}
Every non-wellfounded proof satisfying the SRP can be transformed into a cyclic proof.
\end{theorem}

\begin{proof} Let $\prf$ be an $\calci$-proof of $\lnsg$ which satisfies the SRP. Let $\prf_0$ be the subtree of $\prf$ obtained by pruning every infinite branch of $\prf$ at the lowermost sequent $\lnsh$ such that $\lnsh$ is saturated, there exists a saturated sequent $\lnsk$ strictly below $\lnsh$ with $\size{\lnsk} < \size{\lnsh}$, the path from $\lnsk$ to $\lnsh$ contains a continuous focus trace and passes through the right premise of $\mathsf{r} \in \set{\untilli, \untillii}$ with the formula in focus principal, and $\last{\lnsk}=\last{\lnsh}$. For each such $\lnsh$ we call $\lnsk$ its companion. Since $\prf$ is finitely branching, K\H{o}nig's lemma and Lemma~\ref{lem: suitable repetitions exist} imply that $\prf_0$ is finite and therefore a finite derivation of $\lnsg$. Let $L$ be the set of leaves of $\prf_0$ that were obtained by pruning a branch and let $c$ be the function which maps each such leaf to its companion. Observe that every other leaf which is not contained in $L$ is also a leaf of $\prf$ and therefore an initial sequent. Hence, the tuple $(\pi_{0}, L, c)$ is a cyclic proof of $\lnsg$.
\end{proof}

We obtain the completeness of the cyclic LNS system as a corollary of Theorems~\ref{thm:completeness-non-wf} and~\ref{thm:non-wf-to-cyclic}.

\begin{corollary}[Completeness of $\calccyc$]\label{cor:cyclic-complete}
 If $\lnsg$ is valid, then $\lnsg$ has an $\calccyc$-proof.
\end{corollary}

\subsection{From Cyclic Proofs to Non-Wellfounded Proofs}\label{subsec:unravel}

In this section, we prove the soundness of cyclic proofs by means of a proof transformation. We introduce a technique for unraveling cyclic linear nested sequent proofs; in particular, we show how a cyclic proof rooted at the companion of a cycle can be transformed into a cyclic proof of the corresponding bud (Lemma~\ref{lem:shifting}). This is achieved by shifting rule applications forward from the source proof to the target proof while maintaining a suitable invariant, ensuring that the transformation can be iterated indefinitely. By iterating this shifting procedure, we obtain an infinite sequence of growing cyclic proofs, whose limit is a non-wellfounded derivation. We then show that the limit is in fact a non-wellfounded proof which satisfies the SRP. The main difficulty for the last step is to prove that every infinite branch of the limit derivation is progressing. To show this, we will identify every infinite branch in the limit with an infinite \emph{looping path} of the original cyclic proof. Intuitively, a looping path is a path through a cyclic proof that, upon reaching a bud, may continue at its companion. Due to the use of annotations, we can guarantee that every infinite looping path in a cyclic proof is progressing. This property transfers to the corresponding infinite branch in the limit derivation, thus showing that the limit is a non-wellfounded proof. We begin by introducing looping paths and by showing that every infinite looping path has a suffix with a continuous focus trace. 

\begin{definition}\label{dfn:looping-path}
    Let $(\pi, L, c)$ be a cyclic proof. A \emph{looping path} of $(\pi, L, c)$ is a finite or infinite sequence $\lnspath^{\circ} = \set{\lnsg_i}_{i \in I}$ (where $I$ is an index set) of LNSs in $\pi$ such that the following hold:
    \begin{enumerate}
        \item[$(1)$] $\lnsg_0$ is the root of $\pi$;
        \item[$(2)$] If $\lnsg_i$ is a leaf of $\pi$ and $\lnsg_i \not \in L$, then $\lnspath^\circ$ ends at $\lnsg_i$;
        \item[$(3)$] If $\lnsg_i \in L$, then either $\lnspath^\circ$ ends at $\lnsg_i$ or $\lnsg_{i+1}$ is a child of $c(\lnsg_i)$;
        \item[$(4)$] Otherwise, $\lnsg_{i+1}$ is a child of $\lnsg_i$.
    \end{enumerate}
\end{definition}

If $\lnsg_i \in L$ and $\lnsg_{i+1}$ is a child of $c(\lnsg_i)$, then we say that $\lnspath^\circ$ \emph{passes through the leaf $\lnsg_i$}. Note that if $\lnspath^\circ$ is an infinite looping path, then there exists a leaf $\lnsg_i \in L$ through which $\lnspath^\circ$ passes infinitely often. We denote by $\lnspath^\circ_\infty$ the set of leaves in $L$ through which $\lnspath^\circ$ passes infinitely often. Note that $\lnspath^\circ_\infty$ is finite since $\pi$ is finite. In order to distinguish looping paths from ordinary paths through derivations, we also call the latter \emph{upward paths}. We introduce a pre-order on the set of buds of $\pi$ as follows.

\begin{definition}
    Let $(\pi, L, c)$ be a cyclic proof. The \emph{dependency order} $\preceq$ on $L$ is given by: 
$\lnsg_i \preceq \lnsg_j$ if and only if there exists an upward path from $c(\lnsg_j)$ to $c(\lnsg_i)$. 
A \emph{$\preceq$-maximal leaf} of $L' \subseteq L$ is a leaf $\lnsg_i \in L'$ such that for every leaf $\lnsg_j \in L'$, if $\lnsg_i \preceq \lnsg_j$, then $\lnsg_j \preceq \lnsg_i$.
\end{definition}

The following lemma was proven in~\cite[Lemma 6.10]{Zenger2026}.

\begin{lemma}\label{lem: maximal leaf}
   Let $(\pi, L, c)$ be a cyclic proof and $\lnspath^\circ$ an infinite looping path. Then, every $\preceq$-maximal leaf of $\lnspath^\circ_\infty$ has the same companion.
\end{lemma}

\begin{lemma}\label{lem: stable annotation0}
  Let $(\pi, L, c)$ be a cyclic proof and $\lnspath^\circ$ an infinite looping path. Let $\lnsg_0$ be $\preceq$-maximal in $\lnspath^\circ_\infty$. Then, for every $\lnsg_i \in \lnspath^\circ_\infty$ the upward path from $c(\lnsg_0)$ to $\lnsg_i$ has a continuous focus trace.
\end{lemma}
\begin{proof}
    Recall that by Lemma~\ref{lem: maximal leaf} every $\preceq$-maximal bud in $\lnspath^\circ_\infty$ has the same companion and therefore an upward path from $c(\lnsg_0)$ to $\lnsg_i$ exists for every $\lnsg_i \in \lnspath^\circ_\infty$. We proceed by induction on the cardinality of $\lnspath^\circ_\infty$. The base case where $\lnspath^\circ_\infty = \set{\lnsg_0}$ follows immediately, since by definition the upward path $\lnspath_0$ from $c(\lnsg_0)$ to $\lnsg_0$ has a continuous focus trace. For the inductive step, suppose $\lnspath^\circ_\infty = \set{\lnsg_0,  \ldots, \lnsg_n, \lnsg_{n+1}}$ and the upward path from $c(\lnsg_0)$ to $\lnsg_i$ for $0 \leq i \leq n$ has a continuous focus trace. First, suppose $c(\lnsg_{n+1})$ lies on $\lnspath_0$. By definition the upward path from $c(\lnsg_{n+1})$ to $\lnsg_{n+1}$ has a continuous focus trace. Moreover, the upward path from $c(\lnsg_0)$ to $c(\lnsg_{n+1})$ is an initial segment of $\lnspath_0$ and thus has a continuous focus trace. Hence, the upward path from  $c(\lnsg_0)$ to $\lnsg_{n+1}$ has a continuous focus trace as well. 
    
    Second, suppose $c(\lnsg_{n+1})$ does not lie on $\lnspath_0$. Then, $\lnsg_{n+1}$ is not $\preceq$-maximal by Lemma~\ref{lem: maximal leaf} and our assumption. Therefore, there exists an upward path from $c(\lnsg_0)$ to $c(\lnsg_{n+1})$. Consider the subtree $\pi_{n+1}$ of $\pi$ rooted at $c(\lnsg_{n+1})$. Since $\lnspath^\circ$ passes infinitely often through both $\lnsg_0$ and $\lnsg_{n+1}$, there must exist $1 \leq i \leq n$ such that $\lnsg_i$ belongs to $\pi_{n+1}$ and $c(\lnsg_i)$ does not belong to $\pi_{n+1}$. By induction hypothesis the upward path $\lnspath_i$ from $c(\lnsg_0)$ to $\lnsg_i$ has a continuous focus trace. Note that $\lnspath_i$ passes through $c(\lnsg_{n+1})$. Since the upward path from $c(\lnsg_{n+1})$ to $\lnsg_{n+1}$ has a continuous focus trace, the upward path from $c(\lnsg_0)$ to $\lnsg_{n+1}$ must have a continuous focus trace as well.
\end{proof}

\begin{lemma}\label{lem: stable annotation}
Let $(\pi, L, c)$ be a cyclic proof and $\lnspath^\circ$ an infinite looping path. Then, $\lnspath^\circ$ has a suffix with a continuous focus trace.
\end{lemma}

\begin{proof} Let $\lnspath^\circ = \set{\lnsg_i}_{i \in \mathbb{N}}$ and let $\lnsg$ be $\preceq$-maximal in $\lnspath^\circ_\infty$. Let $j$ be the least natural number such that $\lnsg_j = \lnsg$ and the suffix $\lnspath$ of $\lnspath^\circ$ starting at $\lnsg_j$ only passes through leaves in $\lnspath^\circ_\infty$. By assumption and Lemma~\ref{lem: maximal leaf} every LNS of $\lnspath$ lies on the upward path from $c(\lnsg)$ to some $\lnsh \in \lnspath^\circ_\infty$. By Lemma~\ref{lem: stable annotation0}, every LNS in $\lnspath$ has a formula in focus and so $\lnspath$ has a continuous focus trace.
\end{proof}

Now we turn toward proving the Shifting Lemma (see Lemma~\ref{lem:shifting} below). We begin by introducing a notion of isomorphism between cyclic proofs. Let $\der$ and $\der'$ be derivations in $\calc$. A \emph{homomorphism} from $\der$ to $\der'$ is a function $h$ mapping occurrences of LNSs in $\der$ to occurrences of LNSs in $\der'$ such that whenever a rule $\ru$ occurs in $\der$ with premises $\lnsg_{1},\ldots,\lnsg_{n}$ and conclusion $\lnsg$, the same rule $\ru$ occurs in $\der'$ with premises $h(\lnsg_{1}),\ldots,h(\lnsg_{n})$ and conclusion $h(\lnsg)$.\footnote{Note that homomorphisms map \emph{occurrences} of linear nested sequents, so two distinct occurrences of the same LNS $\lnsg$ may map to different LNSs $\lnsh$ and $\lnsf$.} We say that $\der$ and $\der'$ are \emph{isomorphic} if there exists a homomorphism $h:\der\to\der'$ whose inverse $h^{-1}$ exists and is also a homomorphism, and call $h$ an \emph{isomorphism}. Isomorphisms preserve the structure of derivations, but not necessarily their \emph{content}: isomorphic derivations may be derivations of different LNSs. A \emph{strict isomorphism} is an isomorphism $h: \pi \to \pi'$ which satisfies the following two properties for any $\lnsg$ in $\pi$: (1) $\lnsg$ is saturated iff $h(\lnsg)$ is saturated and (2) $\last{\lnsg} = \last{h(\lnsg)}$; in such a case, we call $\pi$ and $\pi'$ \emph{strictly isomorphic}. We call cyclic proofs $(\pi, L, c)$ and $(\pi', L', c')$ \emph{strictly isomorphic} if there exists a strict isomorphism $\shift: \prf \to \prf'$ such that (1) $\lnsf \in L$ iff $\shift(\lnsf) \in L'$ and (2) for any $\lnsf \in L$, $h(c(\lnsf)) = c'(h(\lnsf))$.

\begin{example} To illustrate the unraveling procedure, we provide an example. Consider the cyclic proof $(\prf,L,c)$ in Figure~\ref{fig:example-cyc-proof} whose conclusion is $c(\lnsg) = \lnsh =  p \until q^\annof \sar p \until q$. As shown in Figure~\ref{fig:example-unravel}, starting from the bud $\lnsg$, we shift the rule applications forward by one component to obtain a proof $(\prf',L',c')$ of $\lnsg$. This procedure can be repeated on the bud $p \sar q \sep p \sar q \sep p \until q^\annof \sar p \until q$, and hence iterated ad infinitum to extract a non-wellfounded proof. Observe that $(\prf,L,c)$ and $(\prf',L',c')$ are strictly isomorphic.
\end{example}

We now introduce notation used in the lemma below. A \emph{context} is an LNS with a \emph{hole} $\nshole$, which replaces a single Gentzen sequent in an LNS (cf.~\cite{Bru09,Lel15}). For example, $\lnsg\nshole = A,B \sar \sep \nshole \sep C \sar D$ is a context. In a context $\lnsg\nshole$, substituting an LNS $\lnsh$ for the hole yields the LNS $\lnsg\set{\lnsh}$.\footnote{Technically, $\lnsg\set{\lnsh}$ is only a well-defined LNS if the conditions placed on the annotations are satisfied. For example if both $\lnsg\nshole$ and $\lnsh$ were to contain a formula in focus, then $\lnsg\set{\lnsh}$ would contain two formulae in focus and thus would not be a well-defined LNS. We ignore this issue here, as this situation does not occur in any of the following proofs.} For example, substituting $\lnsh = E \sar F \sep G \sar$ in the context $\lnsg\nshole$ above yields $\lnsg\set{\lnsh} = A,B \sar \sep E \sar F \sep G \sar \sep C \sar D$. We may write $\lnsg\set{\Gamma \sar \Delta}_{i}$ to indicate that $\Gamma \sar \Delta$ is the $i$-component of $\lnsg$ and $\lnsg\set{\Gamma \sar \Delta \sep \Sigma \sar \Pi}_{i}$ to indicate that $\Gamma \sar \Delta$ is the $i$-component and $\Sigma \sar \Pi$ is the $i+1$-component of $\lnsg$.

\begin{figure}
\begin{center}
\begin{tabular}{c c}
$\prf_{1}' \ :=$
&
\AxiomC{}
\RightLabel{$\id$}
\UnaryInfC{$p \sar q \sep q \sar p, q$}
\AxiomC{}
\RightLabel{$\id$}
\UnaryInfC{$p \sar q \sep q \sar q \sep \sar p \until q$}
\RightLabel{$\untilrii$}
\BinaryInfC{$p \sar q \sep q \sar p \until q$}
\DisplayProof
\end{tabular}
\end{center}

\begin{center}
\begin{tikzpicture}
\node[] (a) [] {\AxiomC{$\prf_{1}'$}

\AxiomC{}
\RightLabel{$\id$}
\UnaryInfC{$p \sar q \sep p \sar p, q \sep p \until q^\annof \sar$}

\AxiomC{$p \sar q \sep p \sar q \sep p \until q^\annof \sar p \until q$}
\RightLabel{$\untilri$}
\BinaryInfC{$p \sar q \sep p \sar p \until q \sep p \until q^\annof \sar$}
\RightLabel{$\untillii$}
\BinaryInfC{$p \sar q \sep p \until q^\annof \sar p \until q$}
\DisplayProof};


\draw[->,rounded corners=.25cm,dashed]
  (2.2,.75) -- (2.2,1) -- (7.2,1) -- node[right,midway]{$c$} 
  (7.2,-.65) -- (.2,-.65);

\end{tikzpicture}
\end{center}

\caption{Shifting Lemma Example.\label{fig:example-unravel}}
\end{figure}

\begin{lemma}[Shifting Lemma]\label{lem:shifting}
Let $\tuple{\der,L,c}$ be a cyclic proof of $\lnsg$ with $\lnsg = c(\lnsg')$ and $\lnsg' \in L$. Then, there exists a strictly isomorphic cyclic proof $\tuple{\der',L',c'}$ of $\lnsg'$.
\end{lemma}

\begin{proof} Let $\lnsg = \lnsh \sslash \Gamma \sar \Delta$ and $\lnsg' = \lnsh \sslash \lnsk \sslash \Gamma \sar \Delta$ such that $\size{\lnsk} \geq 1$, $\size{\lnsg} = n$, and $\size{\lnsg'} = k$. By definition, $\lnsg$ and $\lnsg'$ are saturated. We simultaneously define a derivation $\pi'$ of $\lnsg'$ and a strict isomorphism $\shift$ from LNSs in $\prf$ to LNSs in $\pi'$ in a stepwise manner. We refer to $\prf$ as the `input derivation' and to $\prf'$ as the `output derivation'. We show by induction on the depth $d$ of $\lnsf$ in $\prf$ that the following invariant (I) holds, which is a conjunction of three statements: (a) $h$ is an isomorphism up to depth $d$, (b) $\lnsf$ is saturated \iffi $\shift(\lnsf)$ is saturated, and (c) if $i \geq n$, then $\lnsf(i) = \shift(\lnsf)(j)$ with $j = k + (i - n)$.

We remark that condition (c) ensures that components at or
beyond the index $n$ in LNSs of the input derivation are preserved
via $h$ in the output derivation, but shifted forward by $(k - n)$
positions. This shift reflects the fact that $\lnsg'$ extends $\lnsg$ with the additional components of $\lnsk$. Moreover, since $\lnsg$ is saturated, rule applications in the input derivation are confined to $i$-components with $i \geq n$, and these are precisely the components that are relocated from the index $i$ to the index $k + (i - n)$ in the output.

\case{Base case.} We let the derivation of $\lnsg'$ be the tree consisting of the single node $\lnsg'$ and set $\shift(\lnsg) := \lnsg'$. Note that (I) holds by the shape of $\lnsg$ and $\lnsg'$.


\case{Inductive step.} We simultaneously treat every node $\lnsf$ of depth $d>0$ in $\prf$ and make a case distinction on the rule applied to $\lnsf$ in $\prf$. We show the $\impr$ and $\untilli$ cases as the remaining cases are similar. As noted above, since $\lnsg$ is saturated, all rule applications occur at $i$-components with $i \geq n$.

\case{Case for $\impr$.}Suppose $\impr$ is applied bottom-up (shown below left) to the $i$-component of $\lnsf$ in $\prf$ with $d$ the depth of $\lnsf$ in $\prf$. By IH, we have constructed the output derivation and map $\shift$ up to depth $d$ and know that the $i$-component of $\lnsf$ is equal to the $j$-component of $\shift(\lnsf)$. We extend the output derivation by bottom-up applying $\impr$ to the $j$-component of $\shift(\lnsf)$, as shown below right.
\begin{center}
\begin{tabular}{c c}
\AxiomC{$\lnsf\set{\Sigma, A \sar B, \Pi}_{i}$}
\RightLabel{$\impr$}
\UnaryInfC{$\lnsf\set{\Sigma \sar A \imp B, \Pi}_{i}$}
\DisplayProof

&

\AxiomC{$\shift(\lnsf)\set{\Sigma, A \sar B, \Pi}_{j}$}
\RightLabel{$\impr$}
\UnaryInfC{$\shift(\lnsf)\set{\Sigma \sar A \imp B, \Pi}_{j}$}
\DisplayProof
\end{tabular}
\end{center}
We set $\shift(\lnsf\set{\Sigma, A \sar B, \Pi}_{i}) := \shift(\lnsf)\set{\Sigma, A \sar B, \Pi}_{j}$. Note that $\shift$ is an isomorphism between the input and output derivations up to depth $d+1$, and so (a) is satisfied. Moreover, for all $m \neq i$ with $m \geq n$, the $m$-component has been unaffected by the $\impr$ application, so the $m$-component of $\lnsf\set{\Sigma, A \sar B, \Pi}_{i}$ and the $k+(m-n)$-component of $\shift(\lnsf)\set{\Sigma, A \sar B, \Pi}_{j}$ are equal. Also, observe that the $i$-component of $\lnsf\set{\Sigma, A \sar B, \Pi}_{i}$ is equal to the $j$-component of $\shift(\lnsf)\set{\Sigma, A \sar B, \Pi}_{j}$. This implies that (b) and (c) hold as well and so that (I) holds.

\case{Case for $\untilli$.} Suppose $\untilli$ is applied bottom-up to the $i$-component of $\lnsf$ in $\prf$ with $d$ the depth of $\lnsf$ in $\prf$, where $\anno \in \set{\annou, \annof}$.
\begin{center}
    \AxiomC{$\lnsf\set{\Gamma, B \sar \Delta \sep \Sigma \sar \Pi}_i$}
    \AxiomC{$\lnsf\set{\Gamma, A  \sar \Delta \sep \Sigma, A \until B^\anno \sar \Pi}_i$}
    \BinaryInfC{$\lnsf\set{\Gamma, A \until B^\anno \sar \Delta \sep \Sigma \sar \Pi}_i$}
    \DisplayProof
\end{center}
By IH we have constructed the output derivation and map $\shift$ up to depth $d$ and know that the $i$- and $i+1$-components of $\lnsf$ are equal to the $j$- and $j+1$-components of $\shift(\lnsf)$, respectively. We extend the output derivation by bottom-up applying $\untilli$ to the $j$-component of $\shift(\lnsf)$, as shown below.
\begin{center}
    \AxiomC{$h(\lnsf)\set{\Gamma, B \sar \Delta \sep \Sigma \sar \Pi}_j$}
    \AxiomC{$h(\lnsf)\set{\Gamma, A  \sar \Delta \sep \Sigma, A \until B^\anno \sar \Pi}_j$}
    \BinaryInfC{$h(\lnsf)\set{\Gamma, A \until B^\anno \sar \Delta \sep \Sigma \sar \Pi}_j$}
    \DisplayProof
\end{center}
Note that if $\anno = \annof$, then since $h(\lnsf)\set{\Gamma, A \until B^\annof \sar \Delta \sep \Sigma \sar \Pi}_j$ is a well-defined LNS, so are both premises of the rule instance, implying that the output derivation is well-defined. We set $\shift(\lnsf\set{\Gamma, B \sar \Delta \sep \Sigma \sar \Pi}_i):= h(\lnsf)\set{\Gamma, B \sar \Delta \sep \Sigma \sar \Pi}_j$ and $\shift(\lnsf\set{\Gamma, A  \sar \Delta \sep \Sigma, A \until B^\anno \sar \Pi}_i) := h(\lnsf)\set{\Gamma, A  \sar \Delta \sep \Sigma, A \until B^\anno \sar \Pi}_j$. Note that $h$ is an isomorphism between the input and output derivations up to depth $d+1$, and so (a) is satisfied. Moreover, by an argument similar to the $\impr$-case, one can verify that conditions (b) and (c) are satisfied as well and so that (I) holds. This concludes the construction of the output derivation $\prf'$ and of $\shift$.

We now show how to define a cyclic proof of $\lnsg'$ based on $\prf'$ that is strictly isomorphic to $(\pi, L, c)$. First, observe that (I) implies that $h$ is a strict isomorphism between $\pi$ and $\pi'$. Let $(c(\lnsf),\lnsf)$ form a cycle in $(\prf, L, c)$. By the definition of $\shift$, we know that a path occurs from $\shift(c(\lnsf))$ to $\shift(\lnsf)$ in $\prf'$ with $h(\lnsf)$ a leaf of $\pi'$ and that since $\size{c(\lnsf)} < \size{\lnsf}$, the path from $c(\lnsf)$ to $\lnsf$ passes through an expansion rule and so the path from $h(c(\lnsf))$ to $h(\lnsf)$ passes through an expansion rule as well, implying that $\size{h(c(\lnsf))} < \size{h(\lnsf)}$. Moreover, by the invariant (I), we know that $h(\lnsf)$ and $h(c(\lnsf))$ are saturated since $\lnsf$ and $c(\lnsf)$ are saturated, $\last{\lnsf} = \last{h(\lnsf)}$, and $\last{c(\lnsf)} = \last{h(c(\lnsf))}$. Therefore, since $\last{c(\lnsf)} = \last{\lnsf}$, we have $\last{h(c(\lnsf))} = \last{h(\lnsf)}$. Since $(c(\lnsf), \lnsf)$ is a cycle in a cyclic proof, the path from $c(\lnsf)$ to $\lnsf$ has a continuous focus trace that passes through the right premise of a rule $\ru \in \set{\untilli,\untillii}$ with the formula in focus principal. By construction of $\prf'$, $h$ and by (I), the path from $h(c(\lnsf))$ to $h(\lnsf)$ therefore has a continuous focus trace that passes through the right premise of a rule $\ru \in \set{\untilli,\untillii}$ with the formula in focus principal. Let $L' := \set{h(\lnsf) \mid \lnsf \in L}$ and let $c'$ map each $h(\lnsf) \in L'$ to $h(c(\lnsf))$. For any other leaf $h(\lnsh)$ of $\pi'$, $\lnsh$ is a leaf of $\pi$ which does not belong to $L$ and therefore is an initial sequent. By (I), $h(\lnsh)$ is an initial sequent as well; therefore, $(\pi', L', c')$ is a cyclic proof of $\lnsg'$ that is strictly isomorphic to $(\pi, L,c)$.
\end{proof}

Next, we use Lemma~\ref{lem:shifting} to obtain an infinite sequence of `growing' cyclic proofs and define the limit of this sequence. Let $(\pi_0, L_0, c_0)$ be a cyclic proof of $\lnsg$ with $L_0 \not = \emptyset$. Let $\lnsh_0 = c_{0}(\lnsh_1)$ be a \emph{minimal companion} in $\pi_0$, i.e., a companion node such that no other companion occurs on the path from $\lnsg$ to $\lnsh_{0}$, and let $(\pi, L, c)$ denote the cyclic proof rooted at $\lnsh_0$ in $(\pi_0, L_0, c_0)$.

By Lemma~\ref{lem:shifting}, we can transform $\tuple{\der,L,c}$ into a strictly isomorphic cyclic proof $\tuple{\der',L',c'}$ of $\lnsh_{1}$ such that $\lnsh_{1} = c'(\lnsh_{2})$ for some leaf $\lnsh_{2}$ in $L'$. We construct a new cyclic proof $\tuple{\der_1,L_1,c_1}$ of $\lnsg$ as follows: first, we `paste' $\der'$ above $\lnsh_{1}$ in $\der_0$ to obtain a new derivation $\der_1$. Second, we set $L_1 := (L_0 \setminus \set{\lnsh_{1}}) \cup L'$. Third, we define $c_1$ as follows:
\begin{equation*}
c_1(\lnsf) = 
\begin{cases} 
c_0(\lnsf) & \text{if $\lnsf \in L_0 \setminus \set{\lnsh_{1}}$,} \\
c'(\lnsf) & \text{if $\lnsf \in L'$.}
\end{cases}
\end{equation*}
We first argue that $\tuple{\der_1, L_1, c_1}$ is a cyclic proof of $\lnsg$. Indeed, $\der_1$ is a finite derivation of $\lnsg$ since $\der_0$ is a finite derivation of $\lnsg$ and $\der'$ is a finite derivation of $\lnsh_1$. For condition~(3) of Definition~\ref{dfn:cyclic-proof}, consider first a bud
$\lnsf \in L_0 \setminus \{\lnsh_1\}$. Since the path from the root to $\lnsf$ in $\der_1$ is identical to the corresponding path in $\der_0$, the cycle $(c_0(\lnsf), \lnsf)$ satisfies all required conditions in $\der_1$. Next, consider a bud $\lnsf \in L'$. Since $\tuple{\der', L', c'}$ is a cyclic proof, the path from $c'(\lnsf)$ to $\lnsf$ in $\der'$ is preserved in $\der_1$ and satisfies all required conditions. Finally, every leaf of $\der_1$ not in
$L_1$ is either a leaf of $\der_0$ not in $L_0$ or a leaf of $\der'$ not in $L'$, and is therefore an initial sequent in either case.

Note that the cycle from $\lnsh_{1}$ to $\lnsh_{0}$ in $(\pi_0, L_0, c_0)$ has been shifted upward and replaced by a cycle from $\lnsh_{2}$ to $\lnsh_{1}$ in $\tuple{\der_1,L_1,c_1}$. In general, $\lnsh_0$ may be the companion of several buds, so after repeating the above process for each bud, $\lnsh_0$ is no longer a minimal companion in the generated cyclic proof. By successively repeating the above process for minimal companions, we obtain an infinite sequence $\tuple{\der_0,L_0,c_0}$, $\tuple{\der_1,L_1,c_1}$, $\ldots$ of cyclic proofs. By construction, each $\der_i$ is a tree of LNSs with root $\lnsg$, i.e., $\der_i=(T_i, \leq_i,\lnsg)$ where $T_i$ is a set of nodes (labeled by LNSs) and $(T_i,\leq_i)$ is a poset, such that for each $t \in T_i$, the down-set of $t$ in $T_i$ is well-ordered. Observe that for all $i \in \mathbb{N}$, $T_i \subseteq T_{i+1}$ and $\leq_i \ \subseteq \ \leq_{i+1}$. We define the \emph{limit derivation} of the sequence of cyclic proofs to be $\pi_\omega := (T_\omega, \leq_\omega, \lnsg)$ where $T_\omega := \bigcup_{i \in \mathbb{N}} T_i$ and $\leq_\omega \ := \bigcup_{i \in \mathbb{N}} \leq_i$. Note that $\pi_\omega$ is a well-defined derivation of $\lnsg$ in $\calci$. 

In order to show that every infinite branch $\branch$ of $\pi_\omega$ is progressing, we show that every such $\branch$ \emph{corresponds to} a looping path of $(\pi_0, L_0,c_0)$ in the following sense:

\begin{definition}\label{dfn:corresponds} Given $j,n \in \mathbb{N}$ and finite looping paths $\lnspath_{j}^\circ = \set{\lnsg_i}_{i \leq k}$ of $(\pi_{j}, L_{j}, c_{j})$ and $\lnspath_n^\circ = \set{\lnsh_i}_{i \leq k}$ of $(\pi_n, L_n, c_n)$ of equal length, the looping path $\lnspath^\circ_n$ \emph{corresponds to} $\lnspath^\circ_j$ \iffi for all $i \leq k$:
\begin{enumerate}

\item[$(1)$] $\lnsh_i$ is the conclusion of an instance of $\mathsf{r} \in \calc$ \iffi $\lnsg_i$ is the conclusion of an instance of $\mathsf{r}$;

\item[$(2)$] $\lnsg_i$ is the conclusion of a binary rule $\mathsf{r}$ and $\lnsg_{i+1}$ is the right (resp. left) premise if and only if $\lnsh_i$ is the conclusion of $\mathsf{r}$ and $\lnsh_{i+1}$ is the right (resp. left) 
premise;

\item[$(3)$] $\lnsg_i$ contains a formula in focus \iffi $\lnsh_i$ contains a formula in focus, and if $\lnsg_i$ is the conclusion of a rule $\mathsf{r}$ whose principal formula is in focus, then the principal formula of the rule 
applied to $\lnsh_i$ is in focus (and vice versa);

\item[$(4)$] $\lnsh_i$ is saturated \iffi $\lnsg_i$ is saturated, and if so, $\last{\lnsg_i} = \last{\lnsh_i}$.


\end{enumerate}
An infinite branch $\branch=\set{\lnsh_i}_{i \in \mathbb{N}}$ of $\pi_\omega$ \emph{corresponds to} an infinite looping path $\lnspath^\circ = \set{\lnsg_i}_{i \in \mathbb{N}}$ of $(\pi_0, L_0,c_0)$ \iffi for all $i \in \mathbb{N}$ conditions (1)--(4) are satisfied.
\end{definition}

The proof of the following lemma can be found in the full version of this paper~\cite{LyoZenFull}.

\begin{lemma}\label{cor: corresponding paths} For any $j \in \mathbb{N}$ and every finite looping path $\lnspath^\circ_j$ of $(\pi_{j}, L_{j}, c_{j})$, there exists a finite looping path $\lnspath^\circ_0$ of $(\pi_0, L_0, c_0)$ of equal length, such that $\lnspath^\circ_j$ corresponds to $\lnspath^\circ_0$.
\end{lemma}

We are now ready to prove the main theorem of this section.

\begin{theorem}\label{thm:cyclic-to-non-wf}
Every $\calccyc$-proof can be transformed into an $\calci$-proof satisfying the SRP.
\end{theorem}

\begin{proof} Let $\tuple{\der_0,L_0,c_0}$ be a cyclic proof of $\lnsg$. If $L_0 = \emptyset$, then $\der_0$ is a finite $\calci$-proof of $\lnsg$ and so we are done. Thus, suppose that $L_0 \not = \emptyset$ and let $\pi_\omega$ be the infinite $\calci$-derivation obtained by successively applying the Shifting Lemma to minimal companions as described above. We will now argue that $\pi_\omega$ is a $\calci$-proof. To that end, suppose that $\lnsh$ is a leaf of $\pi_\omega$. Then, there exists an $i \in \mathbb{N}$ such that $\lnsh$ is a leaf of $(\pi_i, L_i, c_i)$. Since $\lnsh$ remains a leaf in $\pi_\omega$, it was not unraveled at any subsequent stage, so $\lnsh \notin L_i$ and hence $\lnsh$ is an initial sequent.

Next, suppose $\branch = \{\lnsg_m\}_{m \in \mathbb{N}}$ is an infinite branch of $\pi_\omega$. Note that by construction of $\pi_\omega$, for every $j \in \mathbb{N}$ there exists a natural number $m_j$ such that $\lnsg_{m_j}$ is a bud of $(\pi_j, L_j, c_j)$ and the initial segment $\branch_j = \set{\lnsg_m}_{m \leq m_j}$ of $\branch$ is a branch of $\pi_j$. We now define an infinite looping path $\lnspath^\circ = \set{\lnsh_i}_{i \in \mathbb{N}}$ of $(\pi_0, L_0, c_0)$ such that $\lnspath^\circ$ corresponds to $\branch$. To do so, we construct an infinite sequence of finite looping paths $\lnsq^\circ_i$ of $(\pi_0, L_0, c_0)$. We proceed by induction on $i$ and show that invariant (II) holds, which is a conjunction of two statements: (i) $\lnsq_{i}^\circ$ contains $\lnsq_{i-1}^\circ$ as an initial segment (for $i > 0$) and (ii) $\lnsq^\circ_i$ corresponds to $\branch_i$.

Let $\lnsq_0^\circ := \branch_0$. By our previous observation $\lnsq_0^\circ$ is a branch of $\pi_0$ ending in a bud. Therefore, (II) holds. Now suppose we have constructed $\lnsq_{i}^\circ$ and (II) holds. Since $\branch_{i+1}$ is a branch of $\pi_{i+1}$, it is also a looping path of $(\pi_{i+1}, L_{i+1}, c_{i+1})$ and therefore corresponds to a looping path $\lnsq^\circ$ of $(\pi_0, L_0, c_0)$ by Lemma~\ref{cor: corresponding paths}. Since $\pi_\omega$ is a tree and $\lnsq^\circ_i$ corresponds to $\branch_i$, which is an initial segment of $\branch_{i+1}$, $\lnsq^\circ_i$ is an initial segment of $\lnsq^\circ$ and so we define $\lnsq^\circ_{i+1} := \lnsq^\circ$. Observe that (II) holds.

Finally, we let $\lnspath^\circ := \bigcup_{i \in \mathbb{N}} \lnsq^\circ_i$.\footnote{We view a sequence $\set{\lnsg_i}_{i \in I}$ of LNSs as a function which maps natural numbers in $I$ to LNSs.} Note that $\lnspath^\circ$ is a well-defined infinite looping path of $(\pi_0, L_0, c_0)$ by (i) and corresponds to $\branch$ by (ii). We now argue that $\branch$ is progressing.

Let $\lnspath^\circ_\infty$ be the set of buds through which $\lnspath^\circ$ passes infinitely often and let $\lnsh \in \lnspath^\circ_\infty$ be $\preceq$-maximal. By Definition~\ref{dfn:cyclic-proof}--(3), the path from $c(\lnsh)$ to $\lnsh$ has a continuous focus trace and passes through the right premise of a rule $\mathsf{r} \in \set{\untilli, \untillii}$ with the formula in focus principal. Since $\lnsh \in \lnspath^\circ_\infty$, $\lnspath^\circ$ passes infinitely often through $\lnsh$ and therefore infinitely often through the right premise of a rule $\mathsf{r} \in \set{\untilli, \untillii}$ with the formula in focus principal. Moreover, by Lemma~\ref{lem: stable annotation}, $\lnspath^\circ$ has a suffix with a continuous focus trace. Therefore, $\lnspath^\circ$ has a progressing suffix.

Since $\lnspath^\circ$ corresponds to $\branch$, we can transfer the progressing property from $\lnspath^\circ$ to $\branch$ using Definition~\ref{dfn:corresponds}. Let $\lnspath^\circ_{\geq j} = \set{\lnsh_i}_{i \geq j}$ be the progressing suffix of $\lnspath^\circ$ and let $\branch_{\geq j} = \set{\lnsg_i}_{i \geq j}$ be the corresponding suffix of $\branch$. By condition~(3), $\lnsg_i$ contains a formula in focus if and only if $\lnsh_i$ contains a formula in focus for each $i \geq j$. Since $\lnspath^\circ_{\geq j}$ has a continuous focus trace, $\branch_{\geq j}$ has a continuous focus trace as well. By conditions~(1) and~(2), whenever $\lnsh_i$ is the conclusion of a rule $\mathsf{r} \in \set{\untilli, \untillii}$ and $\lnsh_{i+1}$ is the right premise, we have that $\lnsg_i$ is the conclusion of the same rule $\mathsf{r}$ and $\lnsg_{i+1}$ is the right premise. Moreover, by condition~(3), if the principal formula of $\mathsf{r}$ applied to $\lnsh_i$ is in focus, then the principal formula of $\mathsf{r}$ applied to $\lnsg_i$ is in focus. Since $\lnspath^\circ$ passes infinitely often through the right premise of $\mathsf{r} \in \set{\untilli, \untillii}$ with the formula in focus principal, the same holds for $\branch$. We thus conclude that $\branch$ is progressing.

Finally, we argue that $\pi_\omega$ satisfies the SRP. By condition~(4), whenever $\lnsh_i$ is saturated, $\lnsg_i$ is saturated as well. Since $\lnspath^\circ$ passes infinitely often through the buds $\lnsh \in \lnspath^\circ_\infty$ and their companions, both of which are saturated by Definition~\ref{dfn:cyclic-proof}--(3), $\lnspath^\circ$ passes through infinitely many saturated LNSs. Therefore, $\branch$ passes through infinitely many saturated LNSs as well. Since this holds for every infinite branch $\branch$ of $\pi_\omega$, $\pi_\omega$ satisfies the SRP. We conclude that $\pi_\omega$ is an $\calci$-proof of $\lnsg$ which satisfies the SRP.
\end{proof}

We obtain the soundness of the cyclic LNS system as a corollary of Theorems~\ref{thm:soundness} and~\ref{thm:cyclic-to-non-wf}. 

\begin{corollary}[Soundness of $\calccyc$]\label{cor:soundness-cyclic-proofs}
If $\lnsg$ has an $\calccyc$-proof, then $\lnsg$ is valid.
\end{corollary}

%% file: sec-conclusion.tex
There are several directions for future research. First, it would be natural to investigate whether---and under what conditions---the proof transformations between non-wellfounded and cyclic proofs developed here can be extended to more expressive multisequent formalisms, such as nested and labeled sequents. Such generalizations would deepen our understanding of cycle recognition and unraveling beyond the linear nested setting. Second, while non-wellfounded Gentzen-style sequent calculi for $\ltl$ were given in~\cite{AloEtAl20,BruLan08,KokStu16}, the problem of syntactic cut-elimination was left open in all three works. We conjecture that the additional structural expressiveness of $\calci$ enables a syntactic cut-elimination proof, which we aim to investigate. Third, it would be interesting to study proof transformations between our linear nested sequent calculi and the existing Gentzen and tableau systems of~\cite{AloEtAl20,BruLan08,KokStu16,Rey16}. Such translations would clarify the relationship between these formalisms and allow results to be transferred between them---in particular, should cut-elimination be established for $\calci$, such 
translations could provide a pathway to resolving the analogous open problem for the Gentzen systems of~\cite{AloEtAl20,BruLan08,KokStu16}. Finally, while our contribution is primarily theoretical, the invertible and formula-driven nature of our cyclic system suggests applications to automated proof-search and proof certificates for $\ltl$.
